\documentclass[sigconf]{aamas}  
\usepackage{balance}  

\usepackage{url}
\usepackage{graphicx}
\usepackage{booktabs}
\usepackage{amsmath}
\usepackage{amsthm}
\usepackage{amsfonts}
\usepackage{dsfont}
\usepackage{amssymb}
\usepackage{algorithm}
\usepackage{algpseudocode}
\usepackage{subcaption}
\usepackage{tabularx}
\usepackage{color}

\newcommand{\tss}[1]{\textsuperscript{#1}}
\newcommand{\canony}{CA\tss{3}NONY }

\theoremstyle{definition}
\theoremstyle{remark}

\DeclareMathOperator*{\argmin}{arg\,\min}

\setcopyright{ifaamas}  
\acmDOI{}  
\acmISBN{}  
\acmConference[AAMAS'19]{Proc.\@ of the 18th International Conference on Autonomous Agents and Multiagent Systems (AAMAS 2019)}{May 13--17, 2019}{Montreal, Canada}{N.~Agmon, M.~E.~Taylor, E.~Elkind, M.~Veloso (eds.)}  
\acmYear{2019}  
\copyrightyear{2019}  
\acmPrice{}  

\begin{document}

\title{Courtesy as a Means to Coordinate}

\author{Panayiotis Danassis}
\affiliation{%
  \institution{\'Ecole Polytechnique F\'ed\'erale de Lausanne (EPFL) \\ Artificial Intelligence Laboratory}
  \city{Lausanne} 
  \country{Switzerland} 
  \postcode{CH-1015}
}
\email{panayiotis.danassis@epfl.ch}
\author{Boi Faltings}
\affiliation{%
  \institution{\'Ecole Polytechnique F\'ed\'erale de Lausanne (EPFL) \\ Artificial Intelligence Laboratory}
  \city{Lausanne} 
  \country{Switzerland} 
  \postcode{CH-1015}
}
\email{boi.faltings@epfl.ch}

\begin{abstract}
	We investigate the problem of multi-agent coordination under rationality constraints. Specifically, role allocation, task assignment, resource allocation, etc. Inspired by human behavior, we propose a framework (CA\tss{3}NONY) that enables fast convergence to efficient and fair allocations based on a simple convention of courtesy. We prove that following such convention induces a strategy which constitutes an $\epsilon$-subgame-perfect equilibrium of the repeated allocation game with discounting. Simulation results highlight the effectiveness of \canony as compared to state-of-the-art bandit algorithms, since it achieves more than two orders of magnitude faster convergence, higher efficiency, fairness, and average payoff.
\end{abstract}



\maketitle

\section{Introduction} \label{Introduction}

In multi-agent systems (MAS), agents are often called upon to implement a joint plan in order to maximize their rewards. Typically, coordination in a joint plan incorporates (possibly a combination of) two distinct elements: agents may be required to take the same action \cite{schelling1960strategy,cooper_1999}, or agents may be required to take distinct actions \cite{grenager2002dispersion,cigler2013decentralized}. The latter is ubiquitous in everyday life, e.g. managing common-pool resources, or deciding on car/ship/airplane routes in traffic management. This paper studies coordination in repeated allocation games. Consider for example the problem of channel allocation in wireless networks \cite{6476077}. In such a problem, $N$ wireless devices contend for $R$ transmission slots (i.e. particular frequency band in a certain period of time), where $N \gg R$. If more than one agent access a slot simultaneously, a collision occurs and the colliding parties incur a cost $\zeta < 0$. The goal then for the agents is to transmit on different slots to minimize collisions over time. Other scenarios include role allocation (e.g. teammates during a game), task assignment (e.g. employees of a factory), resource allocation (e.g. parking spaces and/or charging stations for autonomous vehicles), etc. What follows is applicable to \emph{any such scenario}. For simplicity hereafter we will refer only to indivisible resources (i.e. resources which can not be shared). Beyond the scope of repeated games, the proposed framework can also be applied as a negotiation protocol in one-shot interactions. E.g. self-driving vehicles attempting to get a parking space can utilize such a protocol in a simulated environment with message exchange. Conforming to the relevant literature, in what follows we refer to the coordination problem of selecting distinct actions as the anti-coordination problem \cite{grenager2002dispersion,SSS1817485,cigler2013decentralized,bramoulle2004network,kun2013anti}.

The most straightforward solution would be to have a central coordinator with complete information recommend an action to each agent. This, though, would require the agents to communicate their preferences/plans, which creates high overhead and raises incentive compatibility and truthfulness issues. Moreover, in real-world applications with partial observability, agents might not be willing to trust such recommendations. On the other hand, agents can learn to anti-coordinate their actions in a completely decentralized manner. However, in a fully decentralized scheme, agents have an incentive to 'bully' others (stick to your action until you drive out the competition) \cite{Littman2002}, which makes it impossible to converge to solutions that are both efficient and fair. As such, a mechanism is needed that employs some sort of external authority. In this paper, we employ simple monitoring authorities (MA), with the primary goal to keep track of successful accesses. We do not require any planning capabilities by the MAs, nor any knowledge of the plans and preferences of the participants; only the ability to monitor successful accesses on their respective resource. Such MAs already exist naturally in many domains (e.g. the port authority in maritime traffic management \cite{AAAI1817116}, or the access point in wireless networks, etc.), and can be easily introduced in the rest (e.g. an agent monitoring each charging / parking station, etc.). At the end, agents learn to anti-coordinate their actions in a decentralized manner, while the introduced MAs allow us to impose quotas to the use of resources and punishments upon violating them to align individual incentives with an efficient and fair correlated equilibrium.

A second aspect of anti-coordination problems involves convergence speed. Inter-agent interactions often need to take place in an ad-hoc fashion. Typical approaches (e.g. Monte Carlo algorithms, Bayesian learning, bandit algorithms, etc.) tend to require too many rounds to converge to be feasible in dynamic environments and real-life applications. Yet, humans are able to routinely coordinate in such an ad-hoc fashion. One key concept that facilitates human ad-hoc (anti-)coordination is the use of \emph{conventions} \cite{lewis2008convention}. Behavioral conventions are a fundamental part of human societies, yet they have not appeared meaningfully in empirical modeling of MAS. Inspired by human behavior, we propose the adoption of a simple convention of \emph{courtesy}. Courtesy arises by the social nature of humans. Society demands that an individual should conduct himself in consideration of others. This allows for fast convergence, albeit it is not game theoretically sound; people adhere to it due to social pressure. Such problems become even more severe in situations with scarcity of resources. Under such conditions, courtesy breaks down and in the name of self-preservation people exhibit urgency and competitive behavior \cite{doi:10.1080/09593969.2016.1147476}. Thus, to satisfy our rationality constraint (no incentive to deviate for self-interested agents) in an artificial system, we need a deterrent mechanism.

In this paper we present a framework (CA\tss{3}NONY: \textbf{C}ontextual \textbf{A}nti-coordination in \textbf{A}d-hoc \textbf{Anony}mous games) for repeated allocation games with discounting. \canony reproduces courtesy based on the allocation algorithm of \cite{cigler2013decentralized}, and uses monitoring and punishments as a deterrent mechanism. It exhibits \emph{fast convergence}, and high \emph{efficiency} and \emph{fairness}, while relying only on occupancy feedback which facilitates scalability and does not require any inter-agent communication. The main contributions are:

\begin{itemize}
	\item Introduction of an anti-coordination framework (CA\tss{3}NONY) which consists of a \emph{courteous convention} and a \emph{monitoring scheme}.
	\item Proof that under such a framework, the use of the courteous convention induces strategies that constitute an approximate \emph{subgame-perfect equilibrium}.
	\item Comparison to state-of-the-art bandit algorithms.
\end{itemize}

\section{Related Work} \label{Related Work}

In ad-hoc multi-agent (anti-)coordination the goal is to design autonomous agents that achieve high flexibility and efficiency in a setting that admits no prior coordination between the participants \cite{AAAI10-adhoc}. Typical scenarios include the use of Monte Carlo algorithms \cite{barrett2017making}, Bayesian learning \cite{albrecht2016belief}, or bandit algorithms \cite{ijcai2017-24,barrett2011ad}. Traditionally, pure ad-hoc approaches suffer from slow learning \cite{crandall2014towards}, which makes pure ad-hoc coordination a very ambitious goal for real-life applications. In this paper we propose a middle-ground approach. Inspired by human ad-hoc coordination, we incorporate prior knowledge in the form of simple \emph{conventions}. The coordination can still be considered ad-hoc as it is not pre-programmed, rather it involves learning. This allows for faster convergence compared to pure ad-hoc approaches.

A convention is defined as a customary, expected and self-enforcing behavioral pattern \cite{young1996economics,lewis2008convention}. In MAS, there are two scopes through which we study conventions. First, a convention can be considered as a behavioral rule, designed and agreed upon ahead of time or decided by a central authority \cite{SHOHAM1995231,walker95understandingThe}. Second, a convention may emerge from within the system itself through repeated interactions \cite{mihaylov2014decentralized,walker95understandingThe}. The proposed courteous convention falls on the first category. It is incorporated as prior knowledge, and it is self-enforcing since the induced strategies constitute an approximate subgame-perfect equilibrium of the repeated allocation game.

An alternative way to model the anti-coordination problem is as a multi-armed bandit (MAB) problem \cite{Auer2002}. In MAB problems an agent is given a number of arms and at each time-step has to decide which arm to pull to get the maximum expected reward. Bandit (or no-regret) algorithms typically minimize the total regret of each agent, which is the difference between the expected received payoff and the payoff of the best strategy in hindsight. As such, they satisfy our rationality constraint since they constitute an approximate correlated or coarse correlated equilibrium \cite{nisan2007algorithmic,Roughgarden:2016:TLA:3092750,hart2000simple}. However, the studied problem presents many challenges: there is no stationary distribution (adversarial rewards), all agents are able to learn (similar to recursive modeling) which results to a moving-target problem, and yielding gives a reward of 0 (desirable option for minimizing regret, but not in respect to fairness). Moreover, regret minimization does not necessarily lead to payoff maximization \cite{crandall2014towards}. Nevertheless, due to their ability to learn from partial feedback, bandit algorithms constitute the natural choice for a pure ad-hoc approach. The latter motivates our choice to use them as a baseline, since our agents only receive binary feedback of success or failure upon taking their action.

Game theoretic equilibria are desirable (since they satisfy the rationality constraint), but hard to obtain. Deciding whether an anti-coordination (anonymous) game has a pure Nash equilibrium (NE) is NP-complete \cite{Brandt:2009:SCP:1501035.1501295}. Furthermore, allocation games often admit undesirable equilibria: pure NE which are efficient but not fair, or mixed-strategy NE which are fair but not efficient \cite{cigler2013decentralized}. Hence, iterative best-response algorithms are not satisfactory. On the other hand, correlated equilibria (CE) \cite{AUMANN197467} can be both efficient \& fair, while from a practical perspective they constitute perhaps the most relevant non-cooperative solution concept \cite{hart2000simple}. An optimal CE of an anonymous game may be found in polynomial time \cite{Papadimitriou:2008:CCE:1379759.1379762}. Moreover, it is possible to achieve a CE without a correlation device (central agent) \cite{foster1997calibrated,hart2000simple}, using the history of the actions taken by the opponents. However, we are interested in information-restrictive learning rules (i.e. completely uncoupled \cite{DBLP:journals/corr/Talebi13}), where each agent is only aware of his own history of action/reward pairs. Such an approach was applied in \cite{cigler2013decentralized} to design a decentralized algorithm for reaching efficient \& fair CE in wireless channel allocation games. Yet, while the algorithm reaches an equilibrium in a polynomial number of steps, cooperation to achieve this state is not rational. A self-interested agent could keep accessing a resource forever, until everyone else backs off (also known as `bully' strategy \cite{Littman2002}). In this paper, we build upon the ideas of \citeauthor{cigler2013decentralized} and develop an anti-coordination strategy that constitutes an approximate subgame-perfect equilibrium, i.e. cooperation with the algorithm is a best-response strategy at each sub-game of the original stage game, given any history of the play.

A generalization of anti-coordination games, called dispersion games, was described in \cite{grenager2002dispersion}. In a dispersion game, agents are able to choose from several actions, favoring the one that was chosen by the smallest number of agents (analogous to minority games \cite{challet2013minority}). The authors in \cite{grenager2002dispersion} define a maximal dispersion outcome as an outcome where no agent can switch to an action chosen by fewer agents. The agents themselves do not have any particular preference for the attained equilibrium. Contrary to that, we are interested in achieving an efficient and fair outcome. Besides, in many real world applications, the agents are indifferent to which role/task/resource they attain, as long as they receive one (e.g. wireless frequencies). Tackling dispersion games, and therefore non-binary utilities, remains open for future research.

A similar approach to ours was introduced for weighted matching in \cite{danassis_alma}. The authors propose a heuristic which is decentralized, requires only partial feedback, and has \emph{constant} in the total problem size running time, under reasonable assumptions on the preference domain of the agents. However, it is applicable in cooperative scenarios, and not in the presence of strategic agents.

\section{The \texorpdfstring{\canony}{CA3NONY} framework} \label{canony}

\subsection{The Repeated Allocation Game} \label{The Resource Allocation Problem}

Let a `resource' be any element that can be successfully assigned to only one agent at a time. At each time-step, $\mathcal{N} = \{1, \dots, N\}$ agents try to access $\mathcal{R} = \{1, \dots, R\}$ identical and indivisible resources, where possibly $N \gg R$. The set of available actions is denoted as $\mathcal{A} = \{Y, A_1, \dots, A_R\}$, where $Y$ refers to yielding and $A_r$ refers to accessing resource $r$. We assume that access to a resource is slotted and of equal duration \footnote{This is done to facilitate the proofs. Real world problems have this property anyway (e.g. access to radio channels, role allocation in a
plan, etc.), and the algorithms we compare to all work for slotted resources.}. A successful access yields a positive payoff, while no access has a payoff of 0. If more than one agent access a resource simultaneously, a collision occurs and the colliding parties incur a cost $\zeta < 0$. Thus, the agents only receive a binary feedback of success or failure. Let $a_{n}$ denote agent $n$'s action, and $a_{-n} = \times_{\forall n' \in \mathcal{N} \setminus \{n\}} a_{n'}$ the joint action for the rest of the agents. The payoff function is defined as:

\begin{equation} \label{Eq: payoff function}
	u_n(a_n, a_{-n}) =
	\begin{cases}
		0, & \text{if } a_n = Y \\
		1, & \text{if } a_n \neq Y \land a_i \neq a_n, \forall i \neq n \\
		\zeta, & \text{otherwise} \\
	\end{cases}
\end{equation}

We assume that rewards are discounted by $\delta \in (0, 1)$, and, conforming to real-world scenarios, that each agent $n$ is only aware of his own history of action/reward pairs.

Finally, we assume that the agents can observe side information from their environment at each time-step $t$. We call this side information context (e.g. time, date etc.). The agents utilize this context as a common signal in their decision-making process, a means to learn and anti-coordinate their actions. Let $\mathcal{K} = \{1, \dots, K\}$ denote the context space. The rationale behind the introduction of the common context is that, under completely uncoupled learning rules, having positive probability mass on undesirable actions (e.g. collisions) is unavoidable. Moreover, from a practical perspective, common environmental signals are amply available to the agents \cite{hart2000simple}. We do not assume any a priori relation between the context space and the problem. The only constraints are that the values should repeat periodically, and satisfy $K = \lceil N / R \rceil$.

The context signals could be produced either on the resource side (e.g. by the port authority in maritime traffic management, or the access point in wireless networks), or in a decentralized manner (e.g. in distributed networks with no authorities the senders can attach identifier signals to data traffic \cite{6476077}). Finally, in situations where communication is possible, the agents can agree upon the signal themselves by solving the distributed consensus problem.

\subsection{Adopted Convention} \label{Adopted Convention}

The adopted convention is based on the cooperative allocation algorithm of \cite{cigler2013decentralized}. Each agent $n$ has a strategy $g_n: \mathcal{K} \rightarrow \mathcal{A}$ that determines a resource to access at time-step $t$ after having observed context $k_t$. The strategy is initialized uniformly at random in $\mathcal{A}$. If $g_n(k_t) = A_r$, then agent $n$ accesses resource $r$. Otherwise, if $g_n(k_t) = Y$, the agent does not access a resource but instead chooses uniformly at random a resource $r$ to monitor for activity. If it is free, then the agent updates $g_n(k_t) \leftarrow A_r$ (see Alg. \ref{algo: canony}).

In \cite{cigler2013decentralized}, agents back-off probabilistically in case of a collision (set $g_n(k_t) \leftarrow Y$ with probability $p_{n_{backoff}}$). In such a setting, it is possible to reach a symmetric subgame-perfect equilibrium. But in order to actually play it, the agents need to be able to calculate it. It is not always possible to obtain the closed form of the back-off probability distribution of each resource. Furthermore, a self-interested agent could stubbornly keep accessing a resource forever, until everyone else backs off (`bully' strategy \cite{Littman2002}).

Instead, we adopted a simple convention where agents are being \emph{courteous}, i.e. if there is a collision, the colliding agents will back-off with some constant positive probability: $p_{n_{backoff}} = p > 0, \forall n \in \mathcal{N}$. Being courteous, though, does not satisfy the rationality constraint. However, a uniform distribution of resources is socially optimal (i.e. fair allocations maximize the social welfare). Hence, if we introduce quotas to the resources and punishments upon violating them, courtesy induces rational strategies. In the following sections we introduce a monitoring scheme and prove that the resulting strategy constitutes an $\epsilon$-subgame-perfect equilibrium.

\begin{algorithm}[!t]
	\caption{Learning rule.} \label{algo: canony}
	\begin{algorithmic}[1]
		\Require Initialize $g_n$ u.a.r. in $\mathcal{A}$. Set $accessed_n \leftarrow False$.

		\For{$k_t \in \mathcal{K}$}
			\State Agents observe context $k_t$
			\If{$g_n(k_t) = A_r$ \& $accessed_n = False$}
				\State Agent $n$ accesses resource $r$
				\If{Collision($r$)}
					\State Set $g_n(k_t) \leftarrow Y$ with probability $p_{backoff} > 0$
				\Else
					\State Set $accessed_n \leftarrow True$
				\EndIf
			\ElsIf {$g_n(k_t) = Y$}
				\State Agent $n$ monitors random resource $r \in \mathcal{R}$
				\If{Free($r$)}
					\State Set $g_n(k_t) \leftarrow A_r$
				\EndIf
			\EndIf
		\EndFor
		\State Set $accessed_n \leftarrow False$
	\end{algorithmic}
\end{algorithm}

\subsection{Rationality} \label{Rationality}

In order to ensure the proposed convention's rationality, the agents must be assured that they will eventually be successful, i.e. we must provide safeguards against the monopolization of resources. The proposed framework employs simple Monitoring Authorities (MA). The MAs do not require any planning capabilities, nor any knowledge of the agents' plans and preferences. Their primary goal is to keep track of successful accesses. Depending on the domain, MAs may already exist naturally, or can easily be introduced. Examples include centralized MAs like the port authority in maritime traffic management, or a set of decentralized MAs (one per resource) like the access points in wireless networks, or agents monitoring a charging / parking station. Their function is twofold. First, they could provide occupancy signals. Agents (e.g. \canony, bandit, or Q-learning agents) must be able to receive some form of feedback from their environment to inform collisions and detect free resources. This could be achieved (if possible) by the use of various sensors, or by receiving occupancy signals (e.g. 0, 1) from the MAs. Second, (which is their principal purpose) MAs deter agents from monopolizing resources to the point that each agent can access a resource only for one context value out of $K$. To achieve the latter, MAs must be able to keep track of successful accesses. Upon the violation of the imposed quotas, the framework is responsible to enforce the necessary punishments. Punishments are application specific. They can be individual (e.g. in a wireless scenario, if an agent transmits to some other than the designated channel, then his packets will no longer be relayed), or group punishments (e.g. if quotas are exceeded, access is denied for everyone). The imposed punishments make exceeding the quotas an irrational strategy (simulated in Alg. \ref{algo: canony} by the $accessed_n$ flag), which in turn aligns the individual incentives with an efficient \& fair correlated equilibrium.

\subsubsection{Access Monitoring} \label{Rationality: Access Monitoring}

The primary function of the employed MAs is the tracking of successful accesses by the agents. The latter can be achieved in various ways, depending on the domain. The simplest one would be to maintain a log with successful accesses per episode (period of context values), whether we are dealing with a centralized MA or set of decentralized MAs that are able to communicate to ensure coherence. For more involved scenarios (e.g. if we require agent anonymity per access) we may employ more complex schemes, e.g. using tokens, or artificial currency (solely as an internal mechanism). Note that the latter does not fence the resources (i.e. punishments are still required). In what follows, we present such a self-regulated monitoring scheme.

Initially all the agents that `buy-in' are issued the same amount $m \in \mathbb{R}$ of artificial cash (AC). This amount also corresponds to the initial fee for every resource $f_r \leftarrow m$. To allow access to resource $r$, the MA of $r$ charges $f_r$ units of AC, and monitors the event. If there was a successful access, the MA reimburses the amount of $(1 - \xi) f_r$ AC to the accessing agent, where $\xi \rightarrow 0 \in \mathbb{R}$ is a commission fee. Otherwise, the MA reimburses the full amount of $f_r$ AC to the colliding agents, so that they are able to try again for a different context value. Finally, after each episode, the MAs lower the fee to $f'_r \leftarrow (1 - \xi) f_r, \forall r \in \mathcal{R}$. After every successful access, the amount of AC that an agent possesses drops below the access fee of a resource. Waiting for the fee to drop to the point that $f_r = m / 2$ is not rational since, assuming $\xi \rightarrow 0$, the number of iterations required to allow accessing two resources at the same time will reach $\infty$. At that point the rest of the agents will have reached a correlated equilibrium and the adversarial agent will not have an incentive to access an additional resource, besides the one that corresponds to him, since it would result in a collision. If, due to implementation constraints, we can not select a small $\xi$, the MAs can change the artificial currency every $I$ episodes, invalidating the old one and again making such strategy irrational.

\subsection{Rate of convergence} \label{Rate of convergence}

\begin{theorem} \label{Th: Convergence Steps}
	In a repeated allocation game with $N$ agents and $R$ resources, the expected number of steps before Alg. \ref{algo: canony} converges to a correlated equilibrium is bounded by:

	\begin{equation} \label{Eq: Convergence Steps}
		\mathcal{O}\left( N \left( \log \left\lceil\frac{N}{R}\right\rceil  + 1 \right) \left( \log N + R \right) \right)
	\end{equation}
\end{theorem}

\begin{proof}
	The adopted learning rule is based on the allocation algorithm of \cite{cigler2013decentralized}. For $N, R, K \geq 1$, and back-off probability $0 < p < 1$, the expected number of steps before the algorithm converges is bounded by (\ref{Eq: Convergence Steps Theorem 12}) \footnote{Slightly tighter bound on the convergence speed of the adopted learning rule, based on Theorems 12 and 13 of \cite{cigler2013decentralized}. See the appendix for the proof.}, which for a constant back-off probability and $K = \lceil N / R \rceil$ gives the required bound.

	\begin{equation} \label{Eq: Convergence Steps Theorem 12}
		\mathcal{O}\left( \left( K \log K + 2K \right) R \frac{2 - p}{2 (1 - p)} \left(\frac{1}{p} \log N + R \right)\right)
	\end{equation}
\end{proof}

\begin{corollary} \label{Corollary: optimal back-off probability}
	Under a common, constant back-off probability assumption, $p^* = 2 - \sqrt{2}$ minimizes the convergence time of Alg. \ref{algo: canony} in high congestion scenarios, i.e. $N \gg R$.
\end{corollary}

\begin{proof}
	According to bound (\ref{Eq: Convergence Steps Theorem 12}), in high congestion scenarios (i.e. $\frac{N}{R} = K \rightarrow \infty$), the common, constant back-off probability that minimizes the convergence time is:

	\begin{equation} \label{Eq: optimal back-off probability}
		p^* = \argmin \left( \frac{2 - p}{2 (1 - p)} \frac{1}{p} \right) = 2 - \sqrt{2}
	\end{equation}
\end{proof}

\subsection{Courtesy Pays Off} \label{Courtesy Pays Off}

In this section we prove that if the agents back-off with a constant positive probability $p_{backoff} > 0$, then Alg. \ref{algo: canony} induces a strategy that is an $\epsilon$-equilibrium.

Suppose that in a repeated allocation game with discounting ($\delta \in (0, 1)$) the agents who collide back-off with a constant probability $p_{backoff} > 0$. Let $\sigma^p$ denote the aforementioned strategy (courteous strategy), and $\sigma^*$ denote the optimal (best-response) strategy under the monitoring authorities (possibly better than a stage game NE). Moreover, let $U_n(\sigma, \sigma_{-n}, \delta) = \sum_{t = 0}^{\infty} \delta^t u_n(a_n^t, a_{-n}^t)$ denote the cumulative payoff of agent $n$ following strategy $\sigma$, assuming the rest of the agents follow the strategy $\sigma_{-n}$. The following theorem proves that starting at any time-step, agent $n$ does not gain more than $\epsilon$ by deviating to the optimal strategy $\sigma^*$.

\begin{theorem} \label{Th: epsilon equilibrium}
	Under a high enough discount factor, the courteous strategy $\sigma^p$ constitutes an approximate subgame-perfect equilibrium, i.e. $\forall \epsilon > 0$, $\exists \delta_0 \in (0, 1)$ such that $\forall \delta, \delta_0 \leq \delta < 1$:

	\begin{equation*}
		\mathds{E}[U_n(\sigma_n^p, \sigma_{-n}^p, \delta)] > (1 - \epsilon) \mathds{E}[U_n(\sigma_n^*, \sigma_{-n}^p, \delta)]
	\end{equation*}
\end{theorem}

\begin{proof}
	Note that $\mathds{E}[U_n(\sigma_n^*, \sigma_{-n}^*, \delta)] \geq \mathds{E}[U_n(\sigma_n^*, \sigma_{-n}^p, \delta)]$, since by playing $\sigma_{-n}^p$ the rest of the agents can potentially introduce additional collisions. Thus, it suffices to prove that $\mathds{E}[U_n(\sigma_n^p, \sigma_{-n}^p, \delta)] > (1 - \epsilon) \mathds{E}[U_n(\sigma_n^*, \sigma_{-n}^*, \delta)]$.

	The introduced monitoring scheme prohibits the monopolization of resources, i.e. each agent can only access a resource for his corresponding context value. Thus, the best-response strategy's ($\sigma^*$) payoff for some $\delta$ is bounded by:

	\begin{equation} \label{Eq: expected best-response}
		\mathds{E}[U_n(\sigma_n^*, \sigma_{-n}^*, \delta)] \leq 1 + \delta^K + \delta^{2K} + \dots = \underset{i = 0}{\overset{\infty}{\sum}} \delta ^ {i K} = \frac{1}{1 - \delta^K}
	\end{equation}

	When agents adopt the courteous convention, in each round until the system converges to a correlated equilibrium, the agents receive a payoff between $\zeta < 0$ (collision cost) and 1. Thus, until convergence, the expected payoff is lower bounded by $\zeta \underset{i = 0}{\overset{\tau - 1}{\sum}} \delta ^ {i K} = \zeta \frac{1 - \delta^{\tau K}}{1 - \delta ^ K}$, where $\tau$ is the number of steps to converge. After convergence, their expected payoff is $\underset{i = 0}{\overset{\infty}{\sum}} \delta ^ {\tau + i K} = \frac{\delta^{\tau}}{1 - \delta ^ K}$. Hence, the convention induced strategy's payoff is at least:

	\begin{equation*} \label{Eq: expected intemrediate}
		\mathds{E}[U_n(\sigma_n^p, \sigma_{-n}^p, \delta)] \geq \underset{\tau = 1}{\overset{\infty}{\sum}}Pr[\text{conv. in $\tau$ steps}] \cdot \left( \frac{\zeta (1 - \delta^{\tau K}) + \delta^{\tau}}{1 - \delta ^ K} \right)
	\end{equation*}

	We can define a random variable $X$ such that $X = \tau$ if the algorithm converges after exactly $\tau$ steps. Since $\delta ^ x$ is a convex function we have that $\mathds{E}(\delta ^ x) \geq \delta ^ {\mathds{E}(x)}$, therefore:

	\begin{equation} \label{Eq: expected convention}
		\mathds{E}[U_n(\sigma_n^p, \sigma_{-n}^p, \delta)] \geq \frac{\zeta (1 - \delta^{\mathds{E}(X) K}) + \delta^{\mathds{E}(X)}}{1 - \delta ^ K}
	\end{equation}

	By dividing (\ref{Eq: expected convention}) by (\ref{Eq: expected best-response}) we get:

	\begin{equation} \label{Eq: ratio of expected payoff}
		\frac{\mathds{E}[U_n(\sigma_n^p, \sigma_{-n}^p, \delta)]}{\mathds{E}[U_n(\sigma_n^*, \sigma_{-n}^*, \delta)]} \geq \zeta (1 - \delta^{\mathds{E}(X) K}) + \delta^{\mathds{E}(X)}
	\end{equation}

	$\mathds{E}(X)$ does not depend on $\delta$. Moreover, $\delta^{\mathds{E}(X)}$ is continuous in $\delta$, monotonous, and $\underset{\delta \rightarrow 1^-}{\lim} \delta^{\mathds{E}(X)} = 1$. Thus, we can take the limit of (\ref{Eq: ratio of expected payoff}) as $\delta \rightarrow 1^-$, which equals to $\underset{\delta \rightarrow 1^-}{\lim} \frac{\mathds{E}[U_n(\sigma_n^p, \sigma_{-n}^p, \delta)]}{\mathds{E}[U_n(\sigma_n^*, \sigma_{-n}^*, \delta)]} = 1 $
\end{proof}

In order to guarantee rationality, the discount factor $\delta$ must be close to 1 since, as $\delta$ gets closer to 1, the agents do not care whether they access now or in some future round. Since the proposed monitoring scheme guarantees that every agent will access a resource for his corresponding context value, when $\delta \rightarrow 1$, the expected payoff for agents who are accessing a resource and for those who have not accessed a resource yet will be the same. In other words, the cost (overhead) of learning the correlated equilibrium decreases.

\subsection{Indifference Period} \label{Indifference Period}

In many real world applications, agents are indifferent in claiming a resource in a period of $T_{ind}$ rounds, i.e. $\delta_t = 1, \forall t \leq T_{ind}$. E.g. data of wireless transmitting devices might remain relevant for a specific time-window, during which the agent is indifferent of transmitting. In such cases, we can use the Markov bound to prove that with high probability the proposed algorithm will converge in under $T_{ind}$ time-steps, thus satisfying the rationality constraint. We assume the agents are willing to accept linear `delay' with regard to the number of resources $R$, the number of agents $N$, and the size of the context space $K$, specifically:

\begin{equation} \label{Eq: indifference delay}
	T_{ind} = \mathcal{O}\left( R N K \right)
\end{equation}

\begin{theorem} \label{Th: probability of convergence}
	Under a linear indifference period $T_{ind}$ (i.e. $\delta_t = 1, \forall t \leq T_{ind} = \mathcal{O}\left( R N K \right)$) the probability of the system of agents following Alg. \ref{algo: canony} not having converged during $T_{ind}$ diminishes as the congestion increases ($\frac{N}{R} \rightarrow \infty$). 
\end{theorem}

\begin{proof}
	Using the Markov bound, it follows that the probability that the system takes more than the accepted number of steps ($T_{ind}$) to converge is:

	\begin{equation*} \label{Eq: probability of convergence}
		Pr[\neg \text{conv. after $T_{ind}$}] = \mathcal{O}\left( \frac{\left( \log \left\lceil\frac{N}{R}\right\rceil + 1 \right) \left( \log N + R \right)}{N} \right)
	\end{equation*}
	\noindent
	Taking the limit: $\underset{\frac{N}{R} \rightarrow \infty}{\lim} Pr[\neg \text{conv. after $T_{ind}$}] = 0$.
\end{proof}

Even though the Markov's inequality generally does not give very good bounds when used directly, Th. \ref{Th: probability of convergence} proves that our algorithm converges in the required time with high probability. The latter holds under high congestion, which constitutes the more interesting scenario since for small number of agents the required time to converge is only a few hundreds of time-steps. Moreover, the higher the indifference period, the higher the probability to converge under such time-constraint. For quasilinear indifference period for example, the system converges in the required time-window with high probability even for a small number of agents. We can further strengthen our rationality hypothesis by using a tighter bound (e.g. Chebyshev's inequality), albeit computing the theoretical variance of the convergence time is an arduous task, thus it remains open for future work.

\section{Evaluation} \label{Experimental Evaluation}

In this section we model the resource allocation problem as a multi-armed bandit problem and provide simulation results of CA\tss{3}NONY's performance in comparison to state-of-the-art, well established bandit algorithms, namely the EXP4 \cite{auer2002nonstochastic}, EXP4.P \cite{beygelzimer2011contextual}, and EXP3 \cite{auer2002nonstochastic}. In every case we report the average value over 128 runs of the same simulation. For the EXP family of algorithms, the input parameters are set to their optimal values, as prescribed in the aforementioned publications. We assume a reward of $1$ for a successful access, $-1$ if there is a collision, and $0$ if the agent yielded.

\subsubsection{Level of Courtesy} \label{Level of Courtesy}

We evaluated different back-off probabilities ($p_{backoff} \in \{0.1, 0.25, 2-\sqrt{2}, 0.75, 0.9\}$) for $R = K \in \{2, 4, 8, 16\}$ and $N = R \times K$. There was no significant difference in convergence time, but since $p_{backoff}$ is directly correlated with the number of collisions, it can have significant impact on the average payoff (see Table \ref{table: backoff probability}). It is important to note, though, that agents do not have global nor local knowledge of the level of congestion of the system (they only receive binary occupancy feedback for one resource per time-step). Thus, agents can not select the optimal back-off probability depending on congestion. In what follows, the back-off probability of \canony is set to $p_{backoff} = 2 - \sqrt{2}$ of Eq. \ref{Eq: optimal back-off probability}, which is only optimal under high congestion settings (i.e. $\frac{N}{R} = K \rightarrow \infty$), yet constitutes a safe option. Note that the provided theoretical analysis holds for any constant back-off probability.

\begin{table}[!t]
	\centering
	\caption{Avg. payoff depending on the level of courtesy (back-off probability), $K = R, N = R \times K$.}
	\begin{tabular*}{1\linewidth}{@{\extracolsep{\fill}}rcccc}
	\hline 
	$p_{backoff}$   & $R = 2$ & $R = 4$ & $R = 8$ & $R = 16$ \\ \hline
	$0.1$           & 37.6    & -4.7    & -39.5   & -63.1    \\
	$0.25$          & 43.6    & 10.0    & -16.9   & -40.0    \\
	$2 - \sqrt{2}$  & 45.7    & 15.6    & -5.2    & -23.0    \\
	$0.75$          & 45.5    & 15.7    & -3.7    & -20.6    \\
	$0.9$           & 44.2    & 12.5    & -5.8    & -19.4    \\
	\hline
	\end{tabular*}
	\label{table: backoff probability}
\end{table}

\subsubsection{Bandits \& Monitoring} \label{Bandits & Monitoring}

\canony has three meta actions (Access, Yield, and Monitor $\{A, Y, M\}$), while bandit algorithms have only two ($\{A, Y\}$), and \canony assumes the existence of monitoring authorities (MAs). For fairness' sake in the reported results, we made the following two modifications. First, we include a variation of \canony, denoted as CA\tss{3}NONY*, where we assume agents incur a cost (equal to collision cost $\zeta$) every time they monitor a resource (reflected in Table \ref{table: payoff}). Second, all the employed bandit algorithms make use of the MAs, which indirectly grants them the the ability to monitor resources as well. More specifically, the accumulated payoff is updated only if they are allowed to access a resource. If not, we consider it a monitoring action, which means the agents still receive occupancy feedback. The latter is important, otherwise the bandit algorithms would require significantly longer time to converge. Note that, contrary to CA\tss{3}NONY*, monitoring is free for bandit algorithms with respect to the accumulated payoff, i.e. they do not incur a collision cost.

\subsection{Employed Bandit Algorithms} \label{Employed Bandit Algorithms}

The reward of each arm does not follow a fixed probability distribution (adversarial setting). Moreover, the agents are able to observe side-information (context) at each time-step $t$. The arm that yields the highest expected reward can be different depending on the context. Hence we focus on adversarial contextual bandit algorithms (see \cite{zhou2015survey} for a survey). A typical approach is to use expert advice. In this method we assume a set of experts $\mathcal{M} = \{1, \dots, M\}$ who generate a probability distribution on which arm to pull depending on the context. A no-regret algorithm performs asymptotically as well as the best expert. Such algorithms are the EXP4 and EXP4.P, the difference being that EXP4 exhibits high variance \cite{zhou2015survey}, while EXP4.P achieves the same regret with high probability by combining the confidence bounds of UCB1 \cite{Auer2002} and EXP4. The computational complexity and memory requirements of the above algorithms are linear in $M$, making them intractable for large number of experts. In order to deal with the increased complexity in larger simulations, we gave an edge to these algorithms by restricting the set of experts $\mathcal{M}$ to the same uniform correlated equilibria (CE) that \canony converges to. The latter enabled us to perform larger simulations (the unrestricted versions could not handle $N \geq 64$), while resulting in the same order of magnitude convergence time, slightly higher average payoff, and significantly higher fairness ($> 35\%$ since by design the experts proposed fair CE). Alternatively, we can use non-contextual adversarial bandit algorithms, such as the EXP3. Moreover, we can convert EXP3 to a contextual algorithm by setting up a separate instance for all $k \in \mathcal{K}$. We call this `CEXP3'. This results in a contextual bandit algorithm which has the edge over EXP4 from an implementation viewpoint since its running time at each time-step is $\mathcal{O}(R)$ and its memory requirement is $\mathcal{O}(KR)$ (CA\tss{3}NONY's is $\mathcal{O}(1)$ and $\mathcal{O}(K)$ respectively).

\subsection{Simulation Results} \label{Simulation Results}

\subsubsection{Convergence Speed \& Efficiency} \label{Convergence Speed & Efficiency}

We know that \canony converges to a CE which is efficient. If all agents follow the courteous convention of Alg. \ref{algo: canony}, the system converges to a state where no resources remain un-utilized and there are no collisions (Th. 13 of \cite{cigler2013decentralized}). Furthermore, Th. \ref{Th: Convergence Steps} argues for fast convergence. The former are both corroborated by Fig. \ref{fig: utilization_R4N16K4} (similar results ($> \times 10^2$ faster convergence) were acquired for $R \in \{2, 8, 16\}$ as well.). Fig. \ref{fig: utilization_R4N16K4} depicts the total utilization of resources for a simulation period of $T = 10^6$ time-steps. Note that the $x$-axis is in logarithmic scale. \canony converges significantly ($> \times 10^2$) faster than the bandit algorithms to a state of 100\% efficiency. On the other hand, the bandit algorithms exhibit high variance, never achieve 100\% efficiency, and are not able to handle efficiently the increase in context space size and number of resources.

\begin{figure}[!t]
	\centering
	\includegraphics[width = 1 \linewidth]{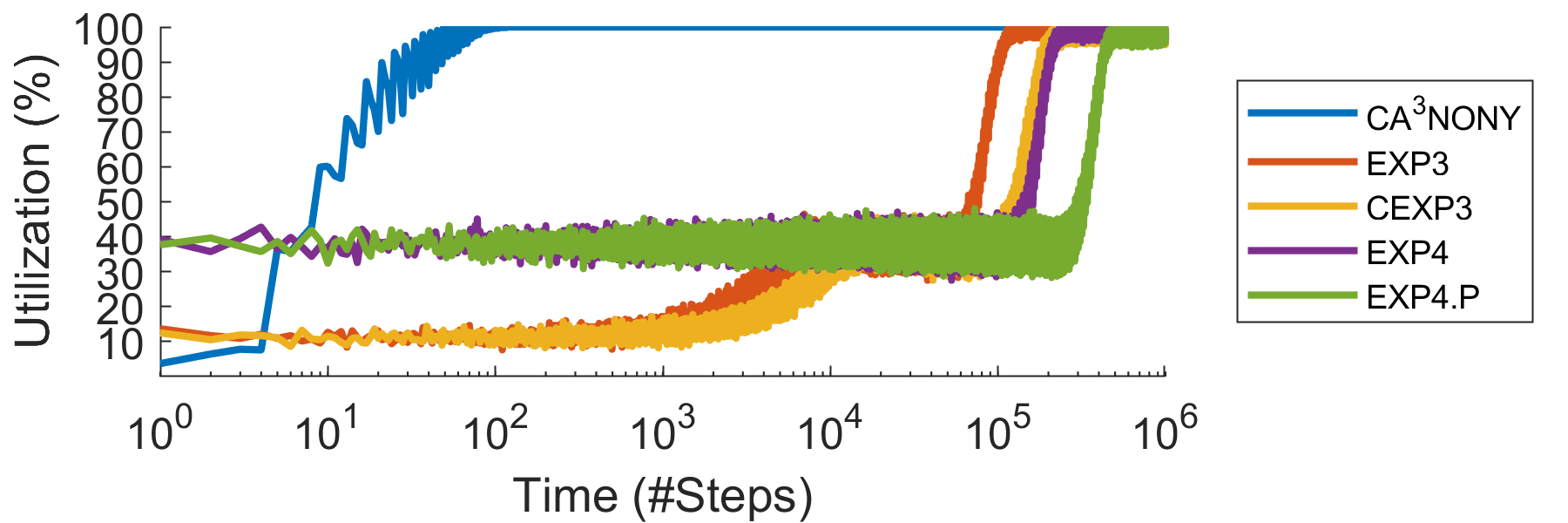}
	\caption{Utilization: $R = 4, N = 16$ ($x$-axis in log scale).}
	\label{fig: utilization_R4N16K4}
\end{figure}

\subsubsection{Fairness} \label{Fairness}

The usual predicament of efficient equilibria for allocation games is that they assign the resources only to a fixed subset of agents, which leads to an unfair result (e.g. an efficient pure NE (PNE) is for $R$ agents to access and $N - R$ agents to yield). This is not the case for \canony, which converges to an equilibrium that is not just efficient but fair as well. Due to the enforced monitoring scheme, all agents acquire the same amount of resources. As a measure of fairness, we will use the Jain index \cite{DBLP:journals/corr/cs-NI-9809099}. The Jain index is widely used by economists, and exhibits a lot of desirable properties such as: population size independence, continuity, scale and metric independence, and boundedness. For an allocation game of $N$ agents, such that the $n^{\text{th}}$ agent receives an allocation of $x_n$, the Jain index is given by $\mathds{J}(\mathbf{x}) = \left|\sum_{n \in \mathcal{N}} x_n\right| ^ 2 \mathbin{/} N \sum_{n \in \mathcal{N}} x_n ^ 2$. An allocation is considered fair, iff $\mathds{J}(\mathbf{x}) = 1$, $\mathbf{x} = (x_1, \dots, x_N) ^\top$. 

Table \ref{table: fairness} presents the expected Jain Index of the evaluated algorithms at the end of the time horizon $T$. \canony converges to a fair equilibrium, achieving a Jain index of 1. The EXP4 and EXP4.P were the fairest amongst the bandit algorithms, achieving a Jain index of close to 1. This is to be expected since the set of experts $\mathcal{M}$ is limited to the same set of equilibria that \canony converges to. On the other hand, EXP3 and CEXP3 performed considerably worse, with the EXP3 exhibiting the worst performance in terms of fairness, equal to a PNE's: $\mathds{J}_{PNE}(\mathbf{x}) = \frac{R^2}{NR} = \frac{1}{K} = \mathds{J}_{EXP3}(\mathbf{x})$.

\subsubsection{Average Payoff} \label{Average Payoff}

The average payoff corresponds to the total discounted payoff an agent would receive in the time horizon $T$. Table \ref{table: payoff} presents the average payoff for the studied algorithms. The clustered pairs of bandit algorithms exhibited $< 1\%$ difference, hence we included the average value of the pair. The discount factor was set to $\delta = 0.99$. At the top half we do not assume any indifference period, while at the bottom half we assume linear indifference period ($\delta_t = 1, \forall t \leq T_{ind}$, where $T_{ind}$ is given by Eq. \ref{Eq: indifference delay}). Recall that \canony and the bandit algorithms do not incur any cost for monitoring resources, while CA\tss{3}NONY* incurs a cost equal to the collision cost $\zeta$.

Once more, \canony (and even CA\tss{3}NONY*) significantly outperforms all the bandit algorithms. The latter have relatively similar performance, with EXP4 and EXP4.P being the best amongst them. It is worth noting that adding an indifference period has a dramatic effect on the results. Comparing \canony to bandit algorithms we observe that \canony achieves a large increase on average payoff, while the opposite happens for the bandit algorithms. This is because the learning rule of Alg. \ref{algo: canony} prohibits from accessing an already claimed resource, thus minimizing collisions. On the contrary, bandit algorithms constantly explore (they assign a positive probability mass to every arm) which leads to collisions. In a multi-agent system where every agent learns this can have a cascading effect. The latter becomes apparent when fixing $\delta = 1$ for $T_{ind}$ steps. The collision cost remains high for longer which, as seen by Table \ref{table: payoff}, has a significant impact on the bandit algorithms' performance.

\begin{table}[!t]
	\centering
	\caption{Fairness (Jain Index), $K = R, N = R \times K$.}
	\begin{tabular*}{1\linewidth}{@{\extracolsep{\fill}}rcccc}
	\hline 
	          & $R = 2$ & $R = 4$ & $R = 8$ & $R = 16$ \\ \hline
	\canony   & 1.0000  & 1.0000  & 1.0000  & 1.0000   \\
	EXP3      & 0.5000  & 0.2500  & 0.1250  & 0.0625   \\
	CEXP3     & 0.7018  & 0.5865  & 0.5341  & 0.9638   \\
	EXP4      & 1.0000  & 0.9999  & 0.9959  & 0.9198   \\
	EXP4.P    & 1.0000  & 0.9999  & 0.9798  & 0.7503   \\
	\hline
	\end{tabular*}
	\label{table: fairness}
\end{table}

\begin{table}[!t]
	\centering
	\caption{Average Payoff, $K = R, N = R \times K$.}
	\begin{tabular*}{1\linewidth}{@{\extracolsep{\fill}}rcccc}
	\hline 
	                & $R = 2$            & $R = 4$                & $R = 8$                    & $R = 16$                     \\ \hline
	\canony         & \phantom{\:\:}45.4 & \phantom{\:\:\:\:}15.7 & \phantom{\:\:\:\:\:\:}-5.3 & \phantom{\:\:\:\:\:\:}-23.0  \\
	CA\tss{3}NONY*  & \phantom{\:\:}20.0 & \phantom{\:\:}-22.4    & \phantom{\:\:\:\:}-50.6    & \phantom{\:\:\:\:\:\:}-72.6  \\
	(C)EXP3         & -72.4              & \phantom{\:\:}-93.1    & \phantom{\:\:\:\:}-99.8    & \phantom{\:\:\:\:}-100.0     \\
	EXP4(.P)        & -59.0              & \phantom{\:\:}-81.0    & \phantom{\:\:\:\:}-90.7    & \phantom{\:\:\:\:\:\:}-95.4  \\ \cline{1-1}
	\canony         & \phantom{\:}53.3   & \phantom{\:\:\:\:}79.0 & \phantom{\:\:\:}502.5      & \phantom{\:\:\:\:}4057.5     \\
	CA\tss{3}NONY*  & \phantom{\:}23.8   & \phantom{\:\:}-55.5    & -1337.1                    & -26723.9                     \\
	(C)EXP3         & -84.0              & -331.0                 & -4175.7                    & -60595.1                     \\
	EXP4(.P)        & -68.4              & -288.4                 & -3807.1                    & -62631.5                     \\
	\hline
	\end{tabular*}
	\label{table: payoff}
\end{table}

\subsubsection{Large Scale Systems} \label{Large Scale Systems}

The innovation of \canony stems from the adoption of a simple convention, which allows its applicability to large scale MAS. To evaluate the latter, Fig. \ref{fig: convergence_time_increasing_R} and \ref{fig: convergence_time_increasing_K} depict the convergence time for increasing number of resources $R$, and increasing system congestion ($\frac{N}{R} = K$) respectively. Both graphs are in a double logarithmic scale, and the error bars represent one standard deviation of uncertainty. The total number of agents is given by $N = R \times K$. Thus, the largest simulations involve $16.384$ agents. Along with CA\tss{3}NONY, we depict the fastest (based on the previous simulations) of the bandit algorithms, namely EXP3. In both cases we acquire $\times 10^3 - \times 10^5$ faster convergence. The above validate CA\tss{3}NONY's performance in both scenarios with abundance ($N \approx R$ or small $K$), and scarcity of resources ($N \gg R$ or large $K$). As depicted, \canony is significantly faster than the EXP3 and can gracefully handle increasing number of resources, and high congestion. Finally note that, in several of the simulations, EXP3 was unable to reach its convergence goal of $90\%$ efficiency (utilization of resources) in a reasonable amount of computation time ($1.5 \times 10^8$ time-steps), hence the resulting gaps in EXP3's lines in Fig. \ref{fig: convergence_time_increasing_R} and \ref{fig: convergence_time_increasing_K}. Especially in situations with scarcity of resources the utilization was significantly lower.

\begin{figure}[!t]
	\centering
	\includegraphics[width = 1 \linewidth]{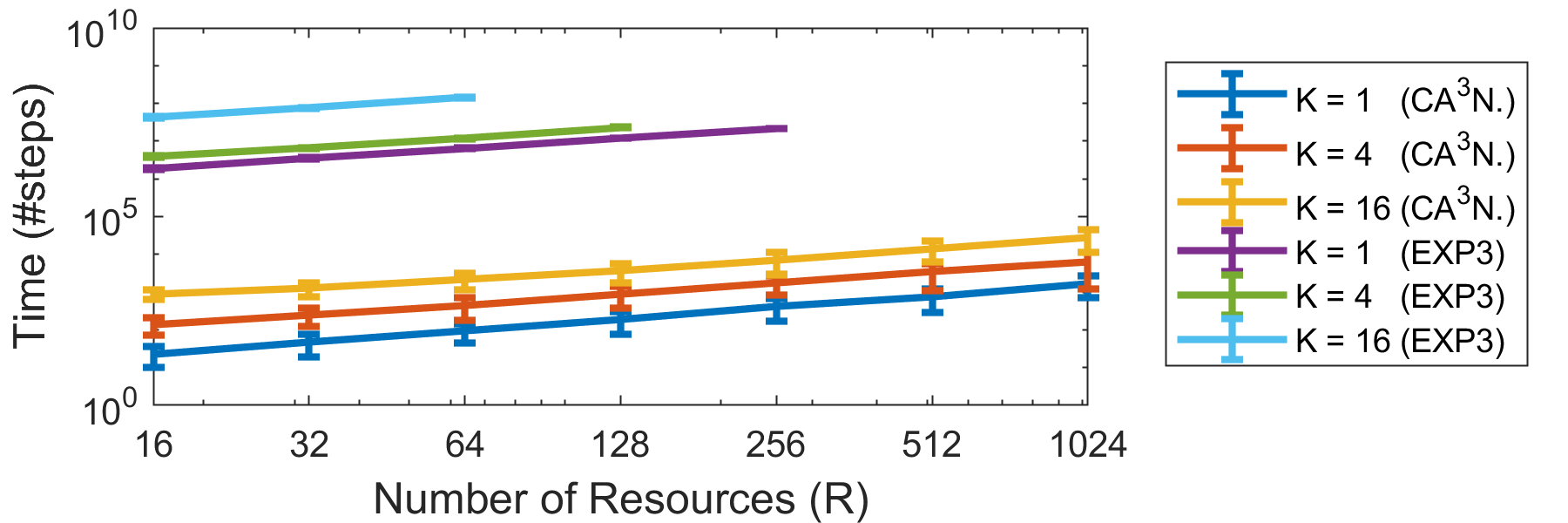}
	\caption{Convergence time: increasing \#resources $R$, varying context space size $K$, $N = R \times K$ (double log scale).}
	\label{fig: convergence_time_increasing_R}
\end{figure}

\begin{figure}[!t]
	\centering
	\includegraphics[width = 1 \linewidth]{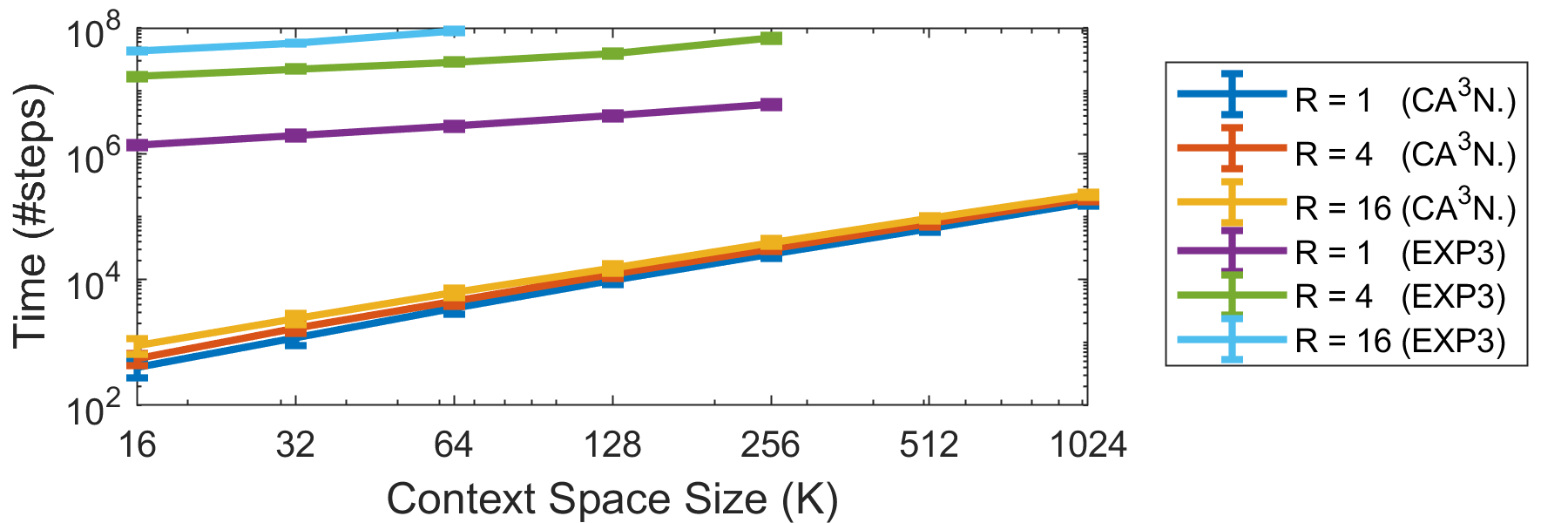}
	\caption{Convergence time: increasing congestion ($\frac{N}{R} = K$), varying \#resources $R$, $N = R \times K$ (double log scale).}
	\label{fig: convergence_time_increasing_K}
\end{figure}

\section{Conclusion} \label{Conclusion}

In this paper we proposed CA\tss{3}NONY, an anti-coordination framework under rationality constraints. It is based on a simple, human-inspired convention of courtesy which prescribes a positive back-off probability in case of a collision. Coupled with a monitoring scheme which deters the monopolization of resources, we proved that the induced strategy constitutes an $\epsilon$-subgame-perfect equilibrium. We compared \canony to state-of-the-art bandit algorithms, namely EXP3, EXP4, and EXP4.P. Simulation results demonstrated that \canony outperforms these algorithms by achieving more than two orders of magnitude faster convergence, a fair allocation, and higher average payoff. The aforementioned gains suggest that human-inspired conventions may prove beneficial in other ad-hoc coordination scenarios as well.

\appendix
\section{Appendix}

\subsection{Proof of Bound (\ref{Eq: Convergence Steps Theorem 12})} \label{Convergence system proof}

In this section we provide a formal proof of bound (\ref{Eq: Convergence Steps Theorem 12}). The proof is an adaptation of the convergence proof of \cite{cigler2013decentralized}. We will prove the following theorem:

\begin{theorem}
	For $N$ agents and $R \geq 1$, $K \geq 1$, $0 < p < 1$ the expected number of steps before the learning algorithm converges to an efficient correlated equilibrium of the allocation game for every $k \in \mathcal{K}$ is $\mathcal{O}\left( \left( K \log K + 2K \right) R \frac{2 - p}{2 (1 - p)} \left(\frac{1}{p} \log N + R \right) +1 \right)$.
\end{theorem}


We begin with the case of having multiple agents but only a single resource (R = 1). We will describe the execution of the proposed learning rule as a discrete time Markov chain (DTMC) \footnote{For an introduction on Markov chains see \cite{norris1998markov}}. In every time-step, each agent performs a Bernoulli trial with probability of `success' $1 - p$ (remain in the competition), and failure $p$ (back-off). When $N$ agents compete for a single resource, a state of the system is a vector $\{0, 1\}^N$ denoting the individual agents that still compete for that resource. But, since the back-off probability is the same for everyone, we are only interested in how many agents are competing and not which ones. Thus, in the single resource case ($R = 1$), we can describe the execution of the proposed algorithm using the following chain:

\begin{definition} \label{def: markov chain X}
	Let $\{X_t\}_{t \geq 0}$ be a DTMC on state space $S = \{0, 1, \ldots ,N\}$ denoting the number of agents still competing for the resource. The transition probabilities are as follows:

	{\small
	\begin{align*}
		& Pr(X_{t + 1} = N | X_t = 0) = 1 & \text{restart} \\
		& Pr(X_{t + 1} = 1 | X_t = 1) = 1 & \text{absorbing} \\
		& Pr(X_{t + 1} = j | X_t = i) = \binom{i}{j} p^{i - j} (1 - p)^j & i > 1, j \leq i \\
	\end{align*}}
	\noindent
	(all the other transition probabilities are zero)
\end{definition}

Intuitively, this Markov chain describes the number of individuals in a decreasing population, but with two caveats: The goal (absorbing state) is to reach a point where only one individual remains, and if we reach zero, we restart.

\begin{lemma}{\cite{cigler2013decentralized}} \label{lm: hitting time of X}
	Let $A = \{0, 1\}$. The expected hitting time of the set of states $A$ in the Markov chain described in Definition \ref{def: markov chain X} is $\mathcal{O}\left( \frac{1}{p} \log N \right)$.
\end{lemma}

\begin{definition} \label{def: markov chain Y}
	Let $\{Y_t\}_{t \geq 0}$ be a DTMC on state space $S = \{0, 1, \ldots ,N\}$ with the following transition probabilities (two absorbing states, 0 and 1):

	{\small
	\begin{align*}
		& Pr(Y_{t + 1} = 0 | Y_t = 0) = 1 & \text{absorbing} \\
		& Pr(Y_{t + 1} = 1 | Y_t = 1) = 1 & \text{absorbing} \\
		& Pr(Y_{t + 1} = j | Y_t = i) = \binom{i}{j} p^{i - j} (1 - p)^j & i > 1, j \leq i \\
	\end{align*}}
	\noindent
	(all the other transition probabilities are zero)
\end{definition}

Let $h_i^A$ denote the hitting probability of a set of states $A$, starting from state $i$. We will prove the following lemma.

\begin{lemma} \label{lm: probability of entering state 1 Y}
	The hitting probability of the absorbing state $\{1\}$, starting from any state $i \geq 1$, of the DTMC $\{Y_t\}$ of Definition \ref{def: markov chain Y} is given by Eq. \ref{Eq: hitting probability bound}. This is a tight lower bound.

	{\small
	\begin{equation} \label{Eq: hitting probability bound}
		h_i^{\{1\}} = \Omega\left( \frac{2(1 - p)}{2 - p} \right), \forall i \geq 1
	\end{equation}}
\end{lemma}

\begin{proof}
	For simplicity we denote $h_i \overset{\Delta}{=} h_i^{\{1\}}$. We will show that for $p \in (0, 1)$, $h_i \geq \lambda = \frac{2(1 - p)}{2 - p}, \forall i \geq 1$ using induction. First note that since state $\{0\}$ is an absorbing state, $h_0 = 0$, $h_1 = 1 \geq \lambda$ and that $\lambda \in (0, 1)$.

	The vector of hitting probabilities $h^A = (h_i^A: i \in S = \{0, 1, \ldots ,N\})$ for a set of states $A$ is the minimal non-negative solution to the system of linear equations \ref{hitting probabilities system}:

	{\small
	\begin{equation} \label{hitting probabilities system}
		\begin{cases}
			h_i^A = 1, &\text{ if } i \in A \\
			h_i^A = \underset{j \in S}{\sum} p_{ij}h_j^A, &\text{ if } i \notin A
		\end{cases}
	\end{equation}}

	By replacing $p_{ij}$ with the probabilities of Definition \ref{def: markov chain Y}, the system of equations \ref{hitting probabilities system} becomes:

	{\small
	\begin{equation} \label{hitting probabilities}
		\begin{cases}
			h_i^A = 1, &\text{ if } i \in A \\
			h_i^A = \underset{j = 0}{\overset{i}{\sum}} \binom{i}{j} p^{i - j} (1 - p)^j h_j^A, &\text{ if } i \notin A
		\end{cases}
	\end{equation}}

	\noindent
	Base case:

	{\small
	\begin{equation*}
		h_2 = (1 - p)^2 h_2 + 2p(1 -p) h_1 + p^2 h_0 = \frac{2p(1 - p)}{1 - (1 - p)^2} = \frac{2(1 - p)}{2 - p}\geq \lambda
	\end{equation*}}

	\noindent
	Inductive step: We assume that $\forall j \leq i - 1 \Rightarrow h_j \geq \lambda$. We will prove that $h_i \geq \lambda, \forall i > 2$.

	{\small
	\begin{align*}
		h_i &= \underset{j = 0}{\overset{i}{\sum}} \binom{i}{j} p^{i - j} (1 - p)^j h_j \\
		&= p^i h_0 + i p^{i - 1} (1 - p) h_1 + \underset{j = 2}{\overset{i - 1}{\sum}} \binom{i}{j} p^{i - j} (1 - p)^j h_j + (1 - p)^i h_i \\
		&\geq p^i h_0 + i p^{i - 1} (1 - p) h_1 + \underset{j = 2}{\overset{i - 1}{\sum}} \binom{i}{j} p^{i - j} (1 - p)^j \lambda + (1 - p)^i h_i \\
		&= i p^{i - 1} (1 - p) + [1 - p^i - (1 - p)^i - i p^{i - 1} (1 - p)] \lambda + (1 - p)^i h_i \\
		\Rightarrow h_i &= \lambda - \frac{p^i}{1 - (1 - p)^i} \lambda + \frac{i p^{i - 1} (1 - p)}{1 - (1 - p)^i} (1 - \lambda)
	\end{align*}}

	\noindent
	We want to prove that $h_i \geq \lambda$:

	{\small
	\begin{align*}
		\frac{i p^{i - 1} (1 - p)}{1 - (1 - p)^i} (1 - \lambda) &\geq \frac{p^i}{1 - (1 - p)^i} \lambda \Rightarrow \\
		\frac{i p^{i - 1} (1 - p) + p^i - p^i}{p^i + i p^{i - 1} (1 - p)} &\geq \lambda \Rightarrow \\
		1 - \frac{p^i}{p^i + i p^{i - 1} (1 - p)} &\geq \frac{2(1 - p)}{2 - p} \Rightarrow \\
		\frac{p^i}{p^i + i p^{i - 1} (1 - p)} &\leq \frac{p}{2 - p} \Rightarrow \\
		p^i (2 - p) &\leq p [p^i + i p^{i - 1} (1 - p)] \Rightarrow \\
		2 - 2p - i + ip &\leq 0 \Rightarrow
		2 - i - p(2 - i) \leq 0 \Rightarrow \\
		(2 - i)(1 - p) &\leq 0 \Rightarrow
		2 - i \leq 0 
	\end{align*}}

	\noindent
	which holds since $i > 2$. The above bound is also tight since $\exists i \in S: h_i = \lambda$, specifically $h_2 = \lambda$.
\end{proof}

Plugging the above hitting probability bound to the convergence proof of \cite{cigler2013decentralized}, results on the required bound.

\bibliographystyle{ACM-Reference-Format}  
\balance  
\bibliography{bibliography}

\end{document}